\documentclass[11pt,fleqn]{article}

\usepackage{amsmath,amssymb,amsthm,enumerate,cite}

\newcommand{\vspacebefore}{\raisebox{0ex}[2.5ex][0ex]{\null}}
\newcommand{\p}{\partial}
\newcommand{\const}{\mathop{\rm const}\nolimits}
\newcommand{\Equiv}{\mathop{ \sim}}

\newcommand{\sign}{\mathop{\rm sign}\nolimits}
\newcommand{\thetbn}{\arabic{nomer}}

\newcommand{\prho}{p}

\newcounter{tbn}

\newcounter{mcasenum}

\newtheorem{theorem}{Theorem}
\newtheorem{lemma}{Lemma}
\newtheorem{corollary}{Corollary}

{\theoremstyle{definition} \newtheorem{definition}{Definition}
\newtheorem{example}{Example}

\newtheorem{note}{Note}
\begin{document}

\par\noindent {\LARGE\bf
Group Analysis of Variable Coefficient\\ Diffusion--Convection Equations.\\ II. Contractions 
 and Exact Solutions
\par}
{\vspace{5mm}\par\noindent {\bf N.M. Ivanova~$^\dag$, R.O. Popovych~$^\ddag$ and C. Sophocleous~$^\S$
} \par\vspace{2mm}\par}
{\vspace{2mm}\par\noindent {\it
$^\dag{}^\ddag$~Institute of Mathematics of NAS of Ukraine, 
3 Tereshchenkivska Str., 01601 Kyiv, Ukraine\\
}}
{\noindent \vspace{2mm}{\it
$\phantom{^\dag{}^\ddag}$~e-mail: ivanova@imath.kiev.ua, rop@imath.kiev.ua
}\par}

{\par\noindent\vspace{2mm} {\it
$^\ddag$~Fakult\"at f\"ur Mathematik, Universit\"at Wien, Nordbergstra{\ss}e 15, A-1090 Wien, Austria
} \par}

{\vspace{2mm}\par\noindent {\it
$^\S$~Department of Mathematics and Statistics, University of Cyprus,
CY 1678 Nicosia, Cyprus\\
}}
{\noindent {\it
$\phantom{^\S}$~e-mail: christod@ucy.ac.cy
} \par}

{\vspace{7mm}\par\noindent\hspace*{8mm}\parbox{140mm}{\small
This is the second part of the series of papers on symmetry properties
of a class of variable coefficient (1+1)-dimensional nonlinear diffusion--convection equations of general form
$f(x)u_t=(g(x)A(u)u_x)_x+h(x)B(u)u_x$.
At first, we review the results of~\cite{Ivanova&Popovych&Sophocleous2006Part1} on equivalence transformations
and group classification of the class under consideration.
Investigation of non-trivial limits of parameterized subclasses of equations
from the given class, which generate contractions of the corresponding maximal Lie invariance algebras,
leads to the natural notion of contractions of systems of differential equations.
After a brief discussion on contractions of symmetries, equations and solutions in general case,
such types of contractions are studied for diffusion--convection equations.
A detailed symmetry analysis of an interesting equation from the class under consideration
is performed.
Exact solutions of some subclasses of the considered class are also given.
}\par\vspace{7mm}}

\section{Introduction}

The present manuscript continues the series of works on symmetry properties
of the nonlinear variable-coefficient diffusion--convection equations of form
\begin{equation} \label{eqDKfgh}
f(x)u_t=(g(x)A(u)u_x)_x+h(x)B(u)u_x,
\end{equation}
where $f=f(x),$ $g=g(x),$ $h=h(x),$ $A=A(u)$ and $B=B(u)$ are arbitrary smooth functions of their variables,
$f(x)g(x)A(u)\!\neq\! 0.$

Based on the results on equivalence transformations and group classification adduced in
the first part~\cite{Ivanova&Popovych&Sophocleous2006Part1} of the presented work,
we continue studying Lie group properties of the equations under consideration.
Namely, investigation of non-trivial limits of parameterized subclasses of equations
from class~(\ref{eqDKfgh}), which generate contractions of the corresponding maximal Lie invariance algebras,
leads us to the natural notion of {\em contraction of (systems of) differential equations}.

It is well-known that exponential cases of parameter-functions,
which admit extensions of maximal Lie invariance algebras,
are  often limits of the power ones with arbitrary exponents~\cite{Bluman&Reid&Kumei1988,Bluman&Kumei1989,Popovych&Ivanova2003PETs}.
By analogy with terminology accepted for the Lie algebras we will call such limits as {\em contractions}.
In such situations some authors exclude exponential cases from formulation of final results of group classification
that is correct only under explicit admission of combined usage of both point equivalence transformations
and contractions in the framework of group classification.
Another famous contraction is one of the 1-dimensional Ricci flow $u_t=\Delta\ln u$ from the porous medium equation
$u_t=m^{-1}\Delta(u^m-1)$ under the limit $m\to0$~\cite{Wu1993}.


All the above-mentioned contractions were found by ad hoc procedures.
To the best of our knowledge this is the first work giving the precise definition and mathematical background
for contractions of equations, algebras of symmetries and solutions. The value of the presented theory is illustrated
in Section~\ref{SectionDKfghContractions} using as example the variable-coefficient diffusion--convection equations~\eqref{eqDKfgh} investigated.

Such contractions allow us to establish additional connections between different cases of extension of maximal Lie invariance algebra
in class~(\ref{eqDKfgh}) and ones between solutions of equations in these cases.
Thus, e.g., applying contractions to the known solutions of initial equations one can easily obtain
exact solutions of the target equations.

In fact, group analysis (in particular, invariant and partially invariant solutions and solutions obtained from separation of variables)
is the only systematic method we know for deducing exact solutions of nonlinear partial differential equations,
and the role of exact solutions in our understanding the mechanics and mathematics of mathematical models cannot be over-estimated.
Classes of invariant solutions usually include self-similar (automodel) solutions arising from scaling invariance,
travelling waves which are invariant with respect to translation symmetries and other classes of solutions having
physical and analytical importance.
For most of the nonlinear systems invariant solutions are the only available exact solutions.

Below we construct exact solutions of equations from class~\eqref{eqDKfgh} using different approaches.
We start by illustration of the simplest approach for deriving exact solutions of the complicated equations
from the ones of a simpler model by means of application of equivalence transformations.
At first sight such investigation may seem to be not very interesting and trivial. However,
a number of authors investigate symmetry properties of different classes of equivalent equations.
In such way they perform a huge number of unnecessary cumbersome calculations.
Since equivalence transformations may be quite complicated, sometimes it seems to be impossible to obtain directly complete and correct results
for some cases which are equivalent to simple ones, although these results can be easily reconstructed with application of equivalence transformations.

At the same time, it is a general mathematical rule that equivalent in some sense objects possess
equivalent properties. In particular, if two systems of equations are equivalent with respect
to point transformations then there exists a one-to-one correspondence between their maximal
Lie invariance algebras, spaces of conservation laws, potential symmetries, exact solutions etc.
All the above mentioned features for more complicated models can be constructed from ones
of the simpler equivalent models by means of application of known equivalence transformation
easier than by direct calculations.

We performed the classical Lie reduction procedure for an equation that cannot be reduced to a constant coefficient equation from class~\eqref{eqDKfgh}.
Then we investigate in more detail an equation having, from the one side, remarkable Lie symmetry property
(namely, $sl(2,\mathbb{R})$-invariance) and, from the other side, admitting application of powerful non-Lie methods of construction of exact solutions.

The present paper is organized as follows.
First of all, in Section~\ref{SectionOnEquivTransGrClasOfEqDKfgh}, for convenience of the reader
we review (in the most suitable for our purpose form) the results of~\cite{Ivanova&Popovych&Sophocleous2006Part1} on
equivalence transformations and group classification of equations~\eqref{eqDKfgh} using simultaneously two different gauges $g=1$ and $g=h$.
Then, in Section~\ref{SectionDKfghContractions} we introduce the notion of contractions of equations
and contractions of symmetry algebras, give a necessary theoretical background
and found contractions of previously classified diffusion--convection equations and those of corresponding symmetry algebras.
Examples of application of additional equivalence transformations are shown in Section~\ref{SectionOnExactSolFromEquiv},
where we reconstruct exact solutions of variable coefficient equations with four-dimensional symmetry algebras
(reducible to constant coefficient ones) from the known solutions of constant coefficient equations.
A case of equation reducible to a constant coefficient form that do not belong to the class under consideration is
presented in Section~\ref{SectionOnExactSolOfGenBurEq}.
In Section~\ref{SectionGrAnOfSl2Eq}
an enhanced group analysis of essentially variable coefficient $sl(2,\mathbb{R})$-invariant diffusion--convection equation is presented
and a set of exact Lie and non-Lie solution is constructed via reduction with respect to the classical Lie and nonclassical symmetries.

Results of Section~\ref{SectionDKfghContractions} will be generalized in the third part~\cite{Ivanova&Popovych&Sophocleous2006Part3}
of this series of papers where we introduce the notion of {\em contractions of conservation laws}, study
the space of local conservation laws and investigate contractions of the conservation laws of equations from class~\eqref{eqDKfgh}.
Ibid we introduce also some generalizations of equivalences of conservation laws that allow us to generalize procedure of construction
of potential systems for systems of differential equations.
Using this new technique we construct all possible inequivalent potential systems for equations from class~\eqref{eqDKfgh}.

\section{Equivalence transformations and group classification\\ of diffusion--convection equations}\label{SectionOnEquivTransGrClasOfEqDKfgh}

In this section we briefly summarize the results of the first part~\cite{Ivanova&Popovych&Sophocleous2006Part1} of this series of papers
concerning equivalence transformations and group classification of class~\eqref{eqDKfgh}, presenting these results
in the most suitable form for our purpose. Due to lack of power of the usual equivalence group~$G^{\sim}$
it seems to be impossible to obtain result of group classification of class~\eqref{eqDKfgh} up to~$G^{\sim}$
in explicit form~\cite{Ivanova&Sophocleous2006}.
So, we adduce the {\it extended equivalence group}~\cite{Ivanova&Popovych&Sophocleous2006Part1,Ivanova&Popovych&Sophocleous2004}
containing transformations depending nonlocally on arbitrary elements.
Using the direct method we construct the complete in this sense group~$\hat G^{\sim}$
formed by the transformations
\begin{gather*}
\tilde t=\delta_1 t+\delta_2,\quad
\tilde x=X(x), \quad
\tilde u=\delta_3 u+\delta_4, \\
\tilde f=\dfrac{\varepsilon_1\delta_1\varphi}{X_x}f, \quad
\tilde g=\varepsilon_1\varepsilon_2^{-1}X_x\varphi\,g, \quad
\tilde h=\varepsilon_1\varepsilon_3^{-1}\varphi\,h, \quad
\tilde A=\varepsilon_2 A, \quad
\tilde B=\varepsilon_3 (B+\varepsilon_4 A),
\end{gather*}
where $\delta_j$ $(j=\overline{1,4})$ and $\varepsilon_i$ $(i=\overline{1,4})$ are arbitrary constants,
$\delta_1\delta_3\varepsilon_1\varepsilon_2\varepsilon_3\not=0$,
$X$ is an arbitrary smooth function of~$x$, $X_x\not=0$,
$\varphi=e^{-\varepsilon_4\int \frac{h(x)}{g(x)}dx}$.

Group~$\hat G^{\sim}$ contains a normal subgroup~$\hat G^{\sim g}$ of
{\em gauge equivalence transformations}~\cite{Ivanova&Popovych&Sophocleous2006Part1,Ivanova&Popovych&Sophocleous2004,LisleDissertation}
\begin{gather}
\tilde f=\varepsilon_1\varphi\, f, \quad
\tilde g=\varepsilon_1\varepsilon_2^{-1}\varphi\, g, \quad
\tilde h=\varepsilon_1\varepsilon_3^{-1}\varphi\, h,\quad
\tilde A=\varepsilon_2 A, \quad
\tilde B=\varepsilon_3 (B+\varepsilon_4 A),
\label{GaugeEquivTransformationsDKfgh}
\end{gather}
where $\varphi=e^{-\varepsilon_4\int \frac{h(x)}{g(x)}dx}$,
$\varepsilon_i$ $(i=\overline{1,4})$ are arbitrary constants, $\varepsilon_1\varepsilon_2\varepsilon_3\not=0$
(the variables $t$, $x$ and $u$ do not transform!).
The transformations~\eqref{GaugeEquivTransformationsDKfgh} act only on arbitrary elements
and do not really change equations.
Application of gauge equivalence transformations is equivalent to rewriting equations
in another form. In spite of really equivalence transformations, their role in group classification
comes not to choice of representatives in equivalence classes but to choice of form of these representatives.

Choice of a gauge for the arbitrary elements is very important for solving and for the final presentation of results.
It is more convenient to constrain the parameter-function~$g$ instead of~$f$ in class~\eqref{eqDKfgh}.
The next problem is choice between gauges of~$g$.
Group classification of~\eqref{eqDKfgh} is performed in~\cite{Ivanova&Popovych&Sophocleous2006Part1} in two different gauges: $g=1$ and $g=h$.
For further investigation we should reduce these results to the floating gauge in order to obtain the simplest
representatives in each equivalence class.
So, below, in contrast to~\cite{Ivanova&Popovych&Sophocleous2006Part1} we adduce the group classification
with floating gauge. The cases $B\not\in\langle 1,A\rangle$ and $A=1$ are investigated in the gauge $g=h$.
In the other cases, to obtain results in a simpler explicit form one should use the gauge $g=1$.

\begin{theorem}\label{TheoremOnClassificationOfeqDKfgh}
The Lie algebra of the kernel of principal groups of~\eqref{eqDKfgh} is $A^{\cap}=\langle \p_t \rangle$.
A complete set of $\hat G^{\Equiv}$-inequivalent equations~\eqref{eqDKfgh} which have
the wider Lie invariance algebras than $A^{\cap}$ is exhausted by cases given in tables~1--3.
\end{theorem}

In tables~1--3 we list all possible $\hat G^{\Equiv}$-inequivalent
sets of functions $f(x),$ $h(x)$ $A(u),$ $B(u)$ and corresponding invariance algebras
under the gauge~$g=1$ or $g=h$.
Numbers with the same Arabic numerals and different Roman letters
correspond to cases that are equivalent with respect to additional equivalence transformations.
Explicit formulas for these transformations can be found in~\cite{Ivanova&Popovych&Sophocleous2006Part1}.
The cases which are contained in tables with different Arabic numbers
or are numbered with different Arabic numerals are reciprocally inequivalent with respect to point transformations,
except the cases with asterisked numeration, presented here to preserve the same numeration as in~\cite{Ivanova&Popovych&Sophocleous2006Part1}.

The operators from tables~1--3 form bases of the maximal invariance algebras
iff the corresponding sets of the functions $f$, $h$, $A$ and $B$ are
$\hat G^{\Equiv}$-inequivalent to ones with more extensive invariance algebras.
For example, in case~3.1 the adduced operators have the above property iff $f\not=f^3$.

Case 1.\ref{gcAaBafexh1}a$'$ is equivalent to case 1.\ref{gcAaBafexh1}a with respect to transformation
\[
\tilde t=t, \quad \tilde x=\ln|x|, \quad \tilde u=u, \quad \tilde A=A, \quad \tilde B=B-A, \quad \tilde p=p+2
\]
from $\hat G^\sim$. We adduce case 1.\ref{gcAaBafexh1}a$'$ here for the convenience of presentation of results only.


\setcounter{tbn}{0}

\begin{center}\footnotesize\renewcommand{\arraystretch}{1.1}
Table~1. Case of $\forall A(u)$ (gauge $g=1$)\\[1ex]
\begin{tabular}{|l|c|c|c|l|}
\hline
N & $B(u)$ & $f(x)$ & $h(x)$ & \hfil Basis of A$^{\rm max}$ \\
\hline
\refstepcounter{tbn}\label{gcAaBafaha}\thetbn & $\forall$ & $\forall$ & $\forall$ & $\p_t$ \\
\refstepcounter{tbn}\label{gcAaBafexh1}\thetbn a&$\forall$ & $e^{p x}$ & 1 &$\p_t,\, p t\p_t+\p_x$ \\
\thetbn a$'$ & $\forall$ & $|x|^p$ & $x^{-1}$ & $\p_t,\, (p+2)t\p_t+x\p_x$ \\
\thetbn b & 1 & $e^x$ & $e^x+\beta$ & $\p_t,\, e^{-t}(\p_t-\p_x)$ \\
\thetbn c & 1 & $|x|^p$ & $x|x|^p+\beta x^{-1}$ & $\p_t,\, e^{-(p+2)t}(\p_t-x\p_x)$ \\
\refstepcounter{tbn}\label{gcAaB1fx-2hlnxx}\thetbn & 1 & $x^{-2}$ & $x^{-1}\ln|x|$ & $\p_t,\, e^{-t}x\p_x$ \\
\refstepcounter{tbn}\label{gcAaB0f1ha}\thetbn a& 0 & 1 & 1 & $\p_t,\, \p_x,\, 2t\p_t+x\p_x$ \\
\thetbn b& 1 & 1 & $x$ & $\p_t,\, e^{-t}\p_x,\, e^{-2t}(\p_t-x\p_x)$ \\
\hline
\end{tabular}
\end{center}
{\footnotesize

Here $p\in\{0,1\}\!\!\mod G^{\sim}_1$ in case~\ref{gcAaBafexh1};
$p\ne-2$ in case \ref{gcAaBafexh1}c;
$\beta\in\{0,\pm 1\}$ in case \ref{gcAaBafexh1}b.

\par}

\vspace{3ex}

\setcounter{tbn}{0}

\begin{center}\footnotesize\renewcommand{\arraystretch}{1.1}
Table~2. Case of $A(u)=e^{\mu u}$ \\[1ex]
\begin{tabular}{|l|c|c|c|c|l|}
\hline
N & $B(u)$ & $f(x)$ & $g(x)$ & $h(x)$ & \hfil Basis of A$^{\rm max}$ \\
\hline
\refstepcounter{tbn}\label{AeuB0faha}\thetbn & 0& $\forall$ & 1 & 1 & $\p_t,\, t\p_t-\p_u$ \\
\refstepcounter{tbn}\label{AeuBenufemxhex}\thetbn & $e^{\nu u}$ & $|x|^p$ & 1 & $|x|^q$ &
$\p_t,\, (p\mu-p\nu-2\nu-q\mu+\mu)t\p_t+(\mu-\nu)x\p_x+(q+1)\p_u$ \\
\thetbn${}^*$ & $e^{\nu u}$ & $e^{p x}$ & 1 & $\varepsilon e^{x}$ &
$\p_t,\, (p\mu-p\nu-\mu)t\p_t+(\mu-\nu)\p_x+\p_u$ \\
\refstepcounter{tbn}\label{AeuBueufex2+xghhex2}\thetbn & $ue^u$ & $e^{p x^2+q x}$ & $e^{p x^2}$ & $e^{p x^2}$ &
$\p_t,\, (2p+q)t\p_t+\p_x-2p\p_u$ \\
\refstepcounter{tbn}\label{AeuBenukf1h1}\thetbn & $e^{\nu u}+\varkappa$ & $1$ & $1$ & 1 &
$\p_t,\, \p_x,\, (\mu-2\nu)t\p_t+((\mu-\nu)x+\nu\varkappa t)\p_x+\p_u$ \\
\refstepcounter{tbn}\label{AeuBuf1h1}\thetbn & $u$ & 1 & 1 & 1 &$\p_t,\, \p_x,\, t\p_t+(x-t)\p_x+\p_u$ \\
\refstepcounter{tbn}\label{AeuB0ff1h}\thetbn a &
0 & $f^1(x)$ & 1 & 1  & $\p_t,\, t\p_t-\p_u,\, \alpha t\p_t+(\beta x^2+\gamma_1x+\gamma_0)\p_x+\beta x\p_u$ \\
\thetbn b &
1 & $|x|^p$ & 1  & $\varepsilon x|x|^{p}$ & $\p_t,\, x\p_x+(p+2)\p_u,\, e^{-\varepsilon(p+2)t}(\p_t-\varepsilon x\p_x)$ \\
\thetbn b${}^*$ &
1 & $e^{x}$ & 1 & $\varepsilon e^{x}$ & $\p_t,\, \p_x+\p_u,\, e^{-\varepsilon t}(\p_t-\varepsilon\p_x)$ \\
\thetbn c & 1 & $x^{-2}$ & 1 & $\varepsilon x^{-1}$ & $\p_t,\, x\p_x,\, t\p_t-\varepsilon tx\p_x-\p_u$ \\
\refstepcounter{tbn}\label{AeuB0f1h}\thetbn a & 0 & 1& 1 & 1 & $\p_t,\, t\p_t-\p_u,\, 2t\p_t+x\p_x,\, \p_x$ \\
\thetbn b & 1 & 1 & 1 & 1 & $\p_t, \,\p_x,\, t\p_t-t\p_x-\p_u,\, 2t\p_t+(x-t)\p_x$ \\
\thetbn c & 1 & $1$ & 1 & $\varepsilon x$ & $\p_t,\, x\p_x+2\p_u,\, e^{-\varepsilon t}\p_x,\,
e^{-2\varepsilon t}(\p_t-\varepsilon x\p_x)$ \\
\thetbn d & 0 & $x^{-3}$ & 1 & 1 & $ \p_t,\, t\p_t-\p_u,\, x\p_x-\p_u,\, x^2\p_x+x\p_u$ \\
\thetbn e & 1 & $x^{-3}$ & 1 & $x^{-2}$ & $\p_t,\, x\p_x-\p_u,\, e^t(\p_t-x\p_x),\, e^t(x^2\p_x+x\p_u)$ \\
\hline
\end{tabular}
\end{center}
{\footnotesize
Here $(\mu,\,\nu)\in\{(0,\,1),\, (1,\,\nu)\}$, $\nu\ne\mu$ in cases~\ref{AeuBenufemxhex}, \ref{AeuBenufemxhex}${}^*$
and~\ref{AeuBenukf1h1};
$\mu=1$ and $\nu\ne1$ in the other cases;
$q\ne-1$ in case~\ref{AeuBenufemxhex}${}^*$ (otherwise it is subcase of the case~1.\ref{gcAaBafexh1}a$'$);
$\varepsilon=\pm1$ in cases~\ref{AeuBenufemxhex}, \ref{AeuB0ff1h}b--\ref{AeuB0ff1h}c and \ref{AeuB0f1h}e;
$p\not\in\{-3,-2,0\}$ in case~\ref{AeuB0ff1h}b;
$\alpha,\beta, \gamma_1, \gamma_0=\const$ and
\[
f^1(x)=\exp\left\{\int\frac{-3\beta x-2\gamma_1+\alpha}
 {\beta x^2+\gamma_1x+\gamma_0}\,dx\right \}.
\]
Case~\ref{AeuBenufemxhex}($q=-1$) is a subcase of case~1.\ref{AeuBenufemxhex}a${}'$.
\par}
\setcounter{tbn}{0}

\begin{center}\footnotesize\renewcommand{\arraystretch}{1.2}
Table~3. Case of $A(u)=|u|^\mu$ \\[1ex]
\begin{tabular}{|l|l|c|c|c|c|l|}
\hline
\hfil N & $\hfil \mu$ & $B(u)$ & $f(x)$ & $g(x)$ & $h(x)$ & \hfil Basis of A$^{\rm max}$ \\
\hline
\refstepcounter{tbn}\label{AumB0fh}\thetbn & $\forall$ & $0$ & $\forall$ & 1 & 1 & $\p_t,\, \mu t\p_t-u\p_u$ \\
\refstepcounter{tbn}\label{AumBunfxlhxg}\thetbn&
$\forall$ & $|u|^\nu$ & $|x|^p$ & 1 & $|x|^q$ & $\p_t,\, (\mu+p\mu-q\mu-p\nu-2\nu)t\p_t$ \\
&&&&&&$+(\mu-\nu)x\p_x+(q+1)u\p_u$ \\
\thetbn${}^*$ &
$\forall$ & $|u|^\nu$ & $e^{px}$ & 1 & $\varepsilon e^{x}$ & $\p_t,\, (p\mu-p\nu-\mu)t\p_t+(\mu-\nu)\p_x+u\p_u$ \\
\refstepcounter{tbn}\label{aumbumlnufeax2+xghheax}\thetbn & $\forall$ & $|u|^\mu\ln|u|$ & $e^{p x^2+q x}$ &
$e^{p x^2}$ & $e^{p x^2}$ & $\p_t,\, (2\mu p+q)t\p_t+\p_x-2p u\p_u$ \\
\refstepcounter{tbn}\label{Au-1B1fx-3ebxhx-2ebx}\thetbn & $\forall$ & 1 & $f^2(x)$ & 1 & $\varepsilon xf^2(x)$ & $\p_t,$ \\
&&&&&& $e^{\varepsilon t}(\p_t-\varepsilon((\mu+1)\beta x^2+x)\p_x-\varepsilon\beta xu\p_u)$ \\
\refstepcounter{tbn}\label{a1bafeax2ghheax2}\thetbn & 0 & $\forall$ & $e^{p x^2}$ &
$e^{p x^2}$ & $e^{p x^2}$ & $\p_t,\, e^{-2p t}\p_x$ \\
\refstepcounter{tbn}\label{a1bafexhghhgex}\thetbn & 0 & $\forall$ & $ e^{x+\gamma e^{x}}$ & $e^{\gamma e^{x}}$ & $e^{\gamma e^{x}}$ &
$\p_t,\, e^{-\gamma t}(\p_t-\gamma \p_x)$ \\
\refstepcounter{tbn}\label{a1bufeax2+xghheax}\thetbn & 0 & $u$ & $e^{p x^2+x}$ &
$e^{p x^2}$ & $e^{p x^2}$ & $\p_t,\, t\p_t+\p_x-2p\p_u$ \\
\refstepcounter{tbn}\label{AumBunkf1h1}\thetbn & $\forall$ & $|u|^\nu+\varkappa$ & $1$ &1 & $1$ & $\p_t,\, \p_x,$ \\
&&&&&& $(\mu-2\nu)t\p_t+((\mu-\nu)x+\nu\varkappa t)\p_x+u\p_u$ \\
\refstepcounter{tbn}\label{AumBlnuf1h1}\thetbn & $\forall$ & $\ln|u|$ & $1$ &1 & $1$ &
$\p_t,\, \p_x,\, \mu t\p_t+(\mu x-t)\p_x+u\p_u$ \\
\refstepcounter{tbn}\label{a1bufeax2ghheax2}\thetbn & 0 & $u$ & $e^{p x^2}$ & $e^{p x^2}$ & $e^{p x^2}$ &
$\p_t,\, e^{-2p t}\p_x,\, \p_x-2p\p_u$ \\
\refstepcounter{tbn}\label{a1blnufeax2ghheax2}\thetbn & 0 & $\ln|u|$ & $e^{p x^2}$ & $e^{p x^2}$ & $e^{p x^2}$ &
$\p_t,\, e^{-2p t}\p_x,\, \p_x-2p u\p_u$ \\
\refstepcounter{tbn}\label{AumB0ff3hh}\thetbn a & $\forall$ & 0 & $f^3(x)$ &1 & 1 &$\p_t,\, \mu t\p_t-u\p_u,$ \\
&&&&&& $\alpha t\p_t+((\mu+1)\beta x^2+\gamma_1x+\gamma_0)\p_x+\beta xu\p_u$ \\
%
%
\thetbn b & $\forall$ & $1$ & $|x|^p$ &1 & $\varepsilon x|x|^{p}$ &
$\p_t,\, \mu x\p_x+(p+2)u\p_u,\, e^{-\varepsilon(p+2)t}(\p_t-\varepsilon x\p_x)$ \\
%
\thetbn b${}^*\!\!\!$ &
$\ne-1$ & 1 & $e^{x}$ &1 & $\varepsilon e^{x}$ & $\p_t,\, \mu\p_x+u\p_u,\, e^{-\varepsilon t}(\p_t-\varepsilon\p_x)$ \\
\thetbn c & $\ne-2$ & 1 & $x^{-2}$ &1 & $\varepsilon x^{-1}$ & $\p_t,\, x\p_x,\, \mu t\p_t-\varepsilon\mu tx\p_x-u\p_u$ \\
\refstepcounter{tbn}\label{A-65B1fx2hx2}\thetbn & $-6/5$ & $1$ & $x^2$ &1 & $x^2$ & $\p_t,\, 2t\p_t+2x\p_x-5u\p_u,$ \\
&&&&&& $t^2\p_t+(2tx+x^2)\p_x-5(t+x)u\p_u$ \\
\refstepcounter{tbn}\label{AumB0f1h}\thetbn a & $\ne -4/3$ & 0 & 1 &1 & 1 &
$\p_t,\, \mu t\p_t-u\p_u,\, \p_x,\, 2t\p_t+x\p_x$\\
\thetbn b & $\ne -4/3$ & 1 &1 & 1 & 1 & $\p_t,\, \mu t\p_t-\mu t\p_x-u\p_u,\, \p_x,\, 2t\p_t+(x-t)\p_x$ \\
\thetbn c & $\ne -4/3$ & 1 & $1$ &1 & $\varepsilon x$ &
$\p_t,\, \mu x\p_x+2u\p_u,\, e^{-\varepsilon t}\p_x,\, e^{-2\varepsilon t}(\p_t-\varepsilon x\p_x)$ \\
\thetbn d & $\ne -4/3, -1$ & 0 & $|x|^{-\frac{3\mu+4}{\mu+1}}$ &1 & 1 &
$\p_t,\, \mu t\p_t-u\p_u,\, (\mu+2)t\p_t-(\mu+1)x\p_x,$ \\ &&&&&& $(\mu+1)x^2\p_x+xu\p_u $ \\
\thetbn e&$\ne -2,$  & 1 & $|x|^{-\frac{3\mu+4}{\mu+1}}$ &1 & $\varepsilon x|x|^{-\frac{3\mu+4}{\mu+1}}$ &
$\p_t,\, \mu(\mu+1)x\p_x-(\mu+2)u\p_u,$ \\
&$-4/3,-1$&&&&& $e^{\varepsilon\frac{\mu+2}{\mu+1}t}(\p_t-\varepsilon x\p_x),\,e^{\varepsilon t}((\mu+1)x^2\p_x+xu\p_u)$\\
\thetbn f & $-1$ & 0 & $e^x$ &1 & 1 & $ \p_t,\, t\p_t+u\p_u,\, \p_x-u\p_u,\,2t\p_t+x\p_x-xu\p_u$ \\
\thetbn g & $-1$ & $1$ & $e^{x}$ &1 & $\varepsilon e^{x}$ &
$\p_t,\, \p_x-u\p_u,\, (x+\varepsilon t-2)\p_x-(x+\varepsilon t)u\p_u,$ \\
&&&&&& $e^{-\varepsilon t}(\p_t-\varepsilon\p_x)$ \\
\thetbn h & $-2$ & $1$ & $x^{-2}$ &1 & $\varepsilon x^{-1}$ & $\p_t,\, x\p_x,\, 2t\p_t-2\varepsilon tx\p_x+u\p_u,$ \\
&&&&&& $e^{\varepsilon t}(x^2\p_x-xu\p_u)$\\
\refstepcounter{tbn}\label{A-43B0f1h}\thetbn a & $-4/3$ & 0 & 1 &1 & 1 &
$\p_t,\, 4t\p_t+3u\p_u,\, \p_x,\, 2t\p_t+x\p_x,$ \\ &&&&&& $x^2\p_x-3xu\p_u$ \\
\thetbn b & $-4/3$ & 1 & 1 &1 & $1$ & $\p_t,\, 4t\p_t+4x\p_x-3u\p_u,\, 2t\p_t+(x-t)\p_x,$ \\
 &&&&&&$\p_x,\, (x+t)^2\p_x-3(x+t)u\p_u$ \\
\thetbn c & $-4/3$ & 1 & 1 &1 & $\varepsilon x$ & $\p_t,\, 2x\p_x-3u\p_u,\, e^{-\varepsilon t}\p_x,$ \\
&&&&&& $e^{-2\varepsilon t}(\p_t-\varepsilon x\p_x),\, e^{\varepsilon t}(x^2\p_x-3xu\p_u)$ \\
\refstepcounter{tbn}\label{A1Buf1h1}\thetbn & 0 & $u$ & 1 &1 & $1$ & $ \p_t,\, \p_x,\, t\p_x-\p_u,\, 2t\p_t+x\p_x-u\p_u,$ \\
&&&&&& $t^2\p_t+tx\p_x-(tu+x)\p_u$ \\
%
\hline
\end{tabular}
\end{center}
{\footnotesize
Here $\nu\ne\mu$; $\varepsilon=\pm1$;
$q\ne-1$ in case~\ref{AumBunfxlhxg}${}^*$ (otherwise it is subcase of the case~1.\ref{gcAaBafexh1}a$'$);
$p\ne-2,-(3\mu+4)/(\mu+1)$ in case~\ref{AumB0ff3hh}c;
$\alpha$, $\beta$, $\gamma_1$, $\gamma_0=\const$,
and
\[
f^2(x)=\exp\left\{\int\frac{-(3\mu+4)\beta x-3}
 {(\mu+1)\beta x^2+x}\,dx\right \},\quad
f^3(x)=\exp\left\{\int\frac{-(3\mu+4)\beta x-2\gamma_1+\alpha}
 {(\mu+1)\beta x^2+\gamma_1x+\gamma_0}\,dx\right \}.
\]
\par}

\section{Contractions of equations and of Lie invariance algebras}\label{SectionDKfghContractions}

In this section we investigate non-trivial limits of parameterized subclasses of equations
from class~(\ref{eqDKfgh}),
which generate contractions of the corresponding maximal Lie invariance algebras
as realizations of abstract Lie algebras in the complete space of (independent and dependent) variables.
By analogy with terminology accepted for abstract Lie algebras,
we will call such limits of parameterized classes and their Lie symmetries as {\it contractions of equations and Lie invariance algebras}
correspondingly.

First we introduce some related notions and
formulate two simple but very useful statements concerning contractions of equations in the general case.

Consider the class~$\{\mathcal{L}^\lambda\}$ of systems~$\mathcal{L}^\lambda$:
$L(x,u_{(\prho)},\lambda)=0$
of $l$~differential equations for $m$~unknown functions $u=(u^1,\ldots,u^m)$
of $n$~independent variables $x=(x_1,\ldots,x_n)$,
which are parameterized with the parameter~$\lambda.$
Here $u_{(\prho)}$ denotes the set of all derivatives of~$u$ with respect to $x$
of order not greater than~$\prho$, including $u$ as the derivatives of the zero order.
$L=(L^1,\ldots,L^l)$ is a tuple of $l$ fixed functions depending on $x,$ $u_{(\prho)}$ and $\lambda$.
For simplicity we assume $\lambda$ as a single numeric (real or complex) parameter.
(Extension to more general case with respect to $\lambda$ is obvious.)
We also suppose that the systems~$\mathcal{L}^\lambda$ are totally nondegenerate.

Let $\hat{\mathcal{L}}^\lambda=\{\hat L^k(x,u_{(\prho)},\lambda)=0,\ k=1,\ldots,\hat l\}$
be a maximal set of algebraically independent differential consequences of $\mathcal{L}^\lambda$
that have, as differential equations, orders not greater than $\prho$.
All such sets of differential consequences determine the same manifold~$\bar{\mathcal{L}}^\lambda$
in the jet space~$J^{(\prho)}$.
Let us fix a value of the parameter $\lambda=\lambda_0$ and
a point~${\bf j}_0\in\bar{\mathcal{L}}^{\lambda_0}$.

\begin{definition}
The systems~$\mathcal{L}^\lambda$ {\it weakly converge} to
the system~$\mathcal{L}^{\lambda_0}$ near the point~${\bf j}_0\in\bar{\mathcal{L}}^{\lambda_0}$
under $\lambda\to\lambda_0$,
iff for any $\lambda$ from a deleted neighborhood~$\Lambda$ of~$\lambda_0$
there exist $\hat{\mathcal{L}}^\lambda$ and
an open neighborhood~$\Omega$ of~${\bf j}_0$ in~$J^{(\prho)}$ such that

1)\ $\hat L^k(x,u_{(\prho)},\lambda)$,  $k=1,\ldots,\hat l$ and their partial derivatives
with respect to $x$ and $u_{(\prho)}$
converge to differential consequences of $\mathcal{L}^{\lambda_0}$ and their corresponding derivatives
pointwise on~$\Omega$ under $\lambda\to\lambda_0$;

2)\ if $\hat L^k$,  $k=1,\ldots,\hat l$ are extended to~$\lambda_0$
by continuity then for an $\hat l$-element subset of the variables~$u_{(\prho)}$ the Jacobian
$|\p\hat L^k/\p v^{k'}|\not=0$ on~$\Omega$ for any~$\lambda$ from a neighborhood of~$\lambda_0$.

We will use the denotation $\mathcal{L}^\lambda\rightharpoonup\mathcal{L}^{\lambda_0}$, $\lambda\to\lambda_0$.
\end{definition}

Consider the parameterized set $\{Q^\lambda,\;\lambda\in\Lambda\}$ of the operators
\[
Q^\lambda=\sum_{i=1}^n\xi^i(x,u,\lambda)\p_{x_i}+\sum_{j=1}^m\eta^j(x,u,\lambda)\p_{u^j}
\]
on the space of variables~$(x,u)$ and a fixed operator
\[
Q^0=\sum_{i=1}^n\xi^{0i}(x,u)\p_{x_i}+\sum_{j=1}^m\eta^{0j}(x,u)\p_{u^j}.
\]
The notation $\mathop{\rm pr}\nolimits_{(\prho)}Q^\lambda\to\mathop{\rm pr}\nolimits_{(\prho)} Q^0$,
$\lambda\to\lambda_0$ denotes that the coefficients of the standard $\prho$-order prolongation
of~$Q^\lambda$ to~$J^{(\prho)}$ pointwise converge to the corresponding coefficients
of the $\prho$-order prolongation of the operator~$Q^0$ on the set~$\Omega$.

\begin{lemma}
Let for any $\lambda$ from a deleted neighborhood of $\lambda_0$ the system~$\mathcal{L}^\lambda$ be Lie invariant
with respect to  an operator~$Q^\lambda$,
$\mathcal{L}^\lambda\rightharpoonup\mathcal{L}^{\lambda_0}$ and
$\mathop{\rm pr}\nolimits_{(\prho)}Q^\lambda\to\mathop{\rm pr}\nolimits_{(\prho)} Q^0$,
$\lambda\to\lambda_0$.
Then the system~$\mathcal{L}^{\lambda_0}$ is Lie invariant with respect to the operator~$Q^0$.
\end{lemma}

Proof directly follows from the infinitesimal invariance criterion.

\begin{corollary}\label{CorollaryOnDifferentiationUnderConstruction}
Suppose that for any $\lambda$ from a deleted neighborhood of $\lambda_0$ the system~$\mathcal{L}^\lambda$ is Lie invariant
with respect to a fixed operator~$Q^0$ and with respect to an operator~$Q^\lambda$
sufficiently smoothly depending on value of~$\lambda$,
$Q^0$ and $Q^\lambda$ are linearly independent,
$\mathcal{L}^\lambda\rightharpoonup\mathcal{L}^{\lambda_0}$ and
$\mathop{\rm pr}\nolimits_{(\prho)}Q^\lambda\to\mathop{\rm pr}\nolimits_{(\prho)} Q^0$,
$\lambda\to\lambda_0$.
Then the system~$\mathcal{L}^{\lambda_0}$ is Lie invariant with respect to the operators~$Q^0$ and
$d^kQ^\lambda/d\lambda^k\bigr|_{\lambda=\lambda_0}$.
Here $d^kQ^\lambda/d\lambda^k$ is the operator with coefficients being $k$-th order derivatives of
the corresponding coefficients of the operator~$Q^\lambda$ with respect to~$\lambda$,
and $k$ is the first integer for which $d^kQ^\lambda/d\lambda^k\bigr|_{\lambda=\lambda_0}$ and $Q^0$ are
linearly independent.
\end{corollary}

\begin{proof}
Both pairs $\tilde Q^\lambda=\tilde Q^0=Q^0$ and
\[
\tilde Q^\lambda=Q^\lambda-
(\lambda-\lambda_0)^{-k}\biggl(
\sum_{i=0}^{k-1}\dfrac{d^iQ^\lambda}{d\lambda^i}\Bigr|_{\lambda=\lambda_0}(\lambda-\lambda_0)^i
\biggr), \qquad
\tilde Q^0=\dfrac{d^kQ^\lambda}{d\lambda^k}\Bigr|_{\lambda=\lambda_0}
\]
satisfy the conditions of the previous lemma.
\end{proof}

\begin{note}
For contracted systems of differential equations we can introduce
the notion of contractions of solutions and then investigate connections between
contractions of Lie invariant solutions and contractions of corresponding subalgebras of
the Lie invariance algebras (see Section~\ref{SectionContrSolutions}).
It is also possible to consider contractions of involutive sets of
reduction operators and solutions which are invariant in the non-classical (conditional) sense.
The study of the above contractions and contractions of symmetries and invariant solutions of
other kinds will be the subject of our further investigations.
\end{note}

There exists a number of different non-trivial contractions of equations from class~(\ref{eqDKfgh})
and, in particular, equations having wider Lie invariance algebras
than the kernel algebra~$A^{\cap}=\langle\p_t\rangle$. Contractions allow us to find
additional connections between different cases of extension of maximal Lie invariance algebra
in class~(\ref{eqDKfgh}) and between solutions of equations in these cases.

The simplest contractions are the ones in parameterized subclasses of equations with respect to their parameters.
For instance, case~2.\ref{AeuB0ff1h}c is the limit of case~2.\ref{AeuB0ff1h}b under $p\to-2$.
The Lie invariance algebra of case~2.\ref{AeuB0ff1h}c
can be obtained from that of case~2.\ref{AeuB0ff1h}b by means of the contraction with $p\to-2$:
\[
Q^{1p}\to Q^1,\quad  Q^{2p}\to Q^2, \quad
-\dfrac{Q^{3p}+\varepsilon Q^{2p}-Q^{1p}}{\varepsilon(p+2)}\to Q^3, \quad p\to-2.
\]
Hereafter the operators are taken from the corresponding tables.
The ``left-hand-side'' (``right-hand-side'') operators are from the given basis of the case
which is the object (target) of the contraction.
The basic operators are numerated in accordance with their order at the tables.

The above contraction illustrates corollary~\ref{CorollaryOnDifferentiationUnderConstruction}
since $Q^{3p}+\varepsilon Q^{2p}\to Q^{1p}=Q^1$, $p\to-2$ and $Q^1=d(Q^{3p}+\varepsilon Q^{2p})/dp\big|_{p=-2}$.
The tuples $(Q^{1p},Q^{2p},Q^{3p})$ and $(Q^1,Q^2,Q^3)$ determine equivalent third-rank realizations
of the algebra~$A_{2.1}\oplus A_1$.
We use the Mubarakzyanov's notations~\cite{Mubarakzyanov1963a} of low-dimensional solvable real Lie algebras.
$A_1$ denotes the one-dimensional Lie algebra and $A_{2.1}$ does the non-Abelian two-dimensional Lie algebra.
It follows from results of~\cite{Popovych&Boyko&Nesterenko&Lutfullin2003} that
up to weak equivalence of realizations there exists the unique third-rank realization
of the algebra~$A_{2.1}\oplus A_1$.
Therefore, equivalence of realizations~$\langle Q^{1p},Q^{2p},Q^{3p}\rangle$ and $\langle Q^1,Q^2,Q^3\rangle$
is foreknown.

In table~3 an analogous contraction is \ref{AumB0ff3hh}.b${}\to{}$\ref{AumB0ff3hh}.c, $p\to-2$, where
\[
Q^{1p}\to Q^1,\quad Q^{2p}\to Q^2, \quad
-\dfrac{\mu Q^{3p}+\varepsilon Q^{2p}-Q^{1p}}{\varepsilon(p+2)}\to Q^3, \quad p\to-2.
\]

Existence of analogous cases of extensions of maximal Lie invariance algebras
and analogous contractions in tables~2 and~3 can be explained in terms of contractions.

Consider equations of the form
\[
f(x)u_t=(g(x)|u|^\mu u_x)_x+h(x)B(u)u_x,
\]
covering the cases from table~3 (or~3$'$), which are parameterized by~$\mu$.
The parameter-function~$B$ takes the values
\[
|u|^\nu,\quad |u|^\nu+\varkappa,\quad 1, \quad 0, \quad |u|^\mu\ln|u|,\quad \ln|u|.
\]
We do not specify here the values of $f(x)$ and $h(x)$ explicitly.
Limits with respect to exponents of~$|u|$ are more cumbersome than the already considered ones.
At first we have to carry out the transformation parameterized with a parameter~$\delta$:
\[
u=1+\frac{\tilde u}\delta, \quad \mu=\delta\tilde\mu, \quad \nu=\delta\tilde\nu,
\quad \alpha=\delta\tilde\alpha,
\quad \gamma_1=\delta\tilde\gamma_1, \quad \gamma_0=\delta\tilde\gamma_0.
\]
Hereafter we change parameters iff they are in the corresponding equation.
For two latter values of~$B$ we have additionally to transform the independent variables and the
parameters~$p$ and~$q$:
\[
t=\delta^2\tilde t,\quad x=\delta\tilde x,\quad
\tilde p=\delta^2 p,\quad \tilde q=\delta q.
\]
Then, we proceed to the limit $\delta\to+\infty$. The limit values are $\bar A=e^{\tilde\mu u}$ and
\[
e^{\tilde\nu u}, \quad e^{\tilde\nu u}+\varkappa, \quad  1, \quad 0, \quad ue^{\tilde\mu u}, \quad u
\]
for $\bar B$ (correspondingly to the order of the above values of~$B$).
For the most of contracted cases the parameter-functions~$f$ and~$h$ do not formally change their values.
The exceptions are the cases~3.\ref{AumB0ff3hh}a, 3.\ref{AumB0f1h}d and 3.\ref{AumB0f1h}e
where $f$ and $h$ depend explicitly on~$\mu$. The described way results in the following contractions:
\begin{gather*}
3.\ref{AumB0fh}_{\mu}\to 2.\ref{AeuB0faha}, \quad
3.\ref{AumBunfxlhxg}_{\mu\nu pq}\to 2.\ref{AeuBenufemxhex}_{\tilde\mu\tilde\nu pq}, \quad
3.\ref{AumBunfxlhxg}^{*}{}_{\mu\nu p}\to 2.\ref{AeuBenufemxhex}^{*}{}_{\tilde\mu\tilde\nu p}, \quad
3.\ref{aumbumlnufeax2+xghheax}_{\mu pq}\to 2.\ref{AeuBueufex2+xghhex2}_{pq}, \quad
3.\ref{AumBunkf1h1}_{\mu\nu pq}\to 2.\ref{AeuBenukf1h1}_{\tilde\nu pq},\!
\\
3.\ref{AumBlnuf1h1}_{\mu}\to 2.\ref{AeuBuf1h1}, \quad
3.\ref{AumB0ff3hh}{\rm a}_{\mu\alpha\beta\gamma_1\gamma_0}\to
   2.\ref{AeuB0ff1h}{\rm a}_{\tilde\alpha\beta\tilde\gamma_1\tilde\gamma_0},\quad
3.\ref{AumB0ff3hh}{\rm b}_{p}\to 2.\ref{AeuB0ff1h}{\rm b}_{\tilde p},\quad
3.\ref{AumB0ff3hh}{\rm b^{*}}\to 2.\ref{AeuB0ff1h}{\rm b^{*}},\quad
\\
3.\ref{AumB0f1h}{\rm a}_{\mu}\to 2.\ref{AeuB0f1h}{\rm a},\quad
3.\ref{AumB0f1h}{\rm b}_{\mu}\to 2.\ref{AeuB0f1h}{\rm b},\quad
3.\ref{AumB0f1h}{\rm c}_{\mu}\to 2.\ref{AeuB0f1h}{\rm c},\quad
3.\ref{AumB0f1h}{\rm d}_{\mu}\to 2.\ref{AeuB0f1h}{\rm d},\quad
3.\ref{AumB0f1h}{\rm e}_{\mu}\to 2.\ref{AeuB0f1h}{\rm e}.
\end{gather*}
Here $\tilde\mu=1$ in all the above cases from table~2 excluding case~2.2.

Analogous contractions can be performed in cases when the parameter-functions~$f$ and~$h$
are powers by means of the transformation
\begin{gather*}
t=\delta^{\mu-2\nu}\tilde t,\quad x=1+\frac{\tilde x}\delta,\quad p=\delta\tilde p,\quad q=\delta\tilde q, \quad
u=\tilde u+\frac1{\nu-\mu}\ln\delta\quad \bigl(\mbox{or} \ u=\delta^{\frac1{\nu-\mu}}\tilde u \bigr)
\end{gather*}
for cases from tables 2 (or 3 correspondingly)
and proceeding to the limit $\delta\to+\infty$:
\begin{gather*}
2.\ref{AumBunfxlhxg}_{\mu\nu pq}\to 2.\ref{AeuBenufemxhex}^{*}{}_{\mu\nu\tilde p,\tilde q=1}, \quad
2.\ref{AeuB0ff1h}{\rm b}_{\tilde p}\to 2.\ref{AeuB0ff1h}{\rm b}^{*}{}_{\tilde p=1},\quad
3.\ref{AumBunfxlhxg}_{\mu\nu pq}\to 3.\ref{AumBunfxlhxg}^{*}{}_{\mu\nu\tilde p,\tilde q=1}, \quad
3.\ref{AumB0ff3hh}{\rm b}_{p}\to 3.\ref{AumB0ff3hh}{\rm b^{*}}.
\end{gather*}
To obtain similar contractions of cases from table~1, we have to use another transformation
of the variables and parameters:
\[
t=\dfrac{\tilde t}\delta,\quad x=1+\frac{\tilde x}\delta,\quad u=\tilde u,\quad
p=\delta\tilde p,\quad A=\dfrac{\tilde A}\delta
\]
that leads to the contractions
$1.\ref{gcAaBafexh1}{\rm a}'{}_{pA}\to 1.\ref{gcAaBafexh1}{\rm a}_{\tilde p\tilde A}$,
$1.\ref{gcAaBafexh1}{\rm c}_{p\beta A}\to 1.\ref{gcAaBafexh1}{\rm b}_{\tilde p\beta\tilde A}$.

The question whether the above contractions exhaust non-trivial ones between equations from class~(\ref{eqDKfgh}),
which admit extensions of maximal invariance algebras, remains open.

\section{Construction of exact solutions\\ via equivalence transformations}\label{SectionOnExactSolFromEquiv}

Results of group classification and additional equivalence transformations can be used for construction of exact solutions of equations
from class~\eqref{eqDKfgh} admitting extensions of maximal Lie invariance group.
In~\cite{Popovych&Ivanova2004NVCDCEs,Ivanova&Sophocleous2006} known exact solutions of ``constant coefficient''
diffusion--convection equations were mapped
to ones of variable coefficient equations~\eqref{eqDKfgh} 
by means of additional equivalence transformations.

In view of Theorem~5 of~\cite{Ivanova&Popovych&Sophocleous2006Part1}
if an equation of form~\eqref{eqDKfgh} is invariant with respect to a Lie algebra
of dimension not less than 4 then it can be reduced by point transformations to one with $f=g=h=1$.
Therefore, we can reconstruct solutions of these variable-coefficient equations from the already known solutions
using additional equivalence transformations found in~\cite{Ivanova&Popovych&Sophocleous2006Part1}.
Thus, e.g., equation~3.\ref{AumB0f1h}g
\begin{equation}\label{EqToSolve14g}
e^xu_t=(u^{-1}u_x)_x+\varepsilon e^xu_x
\end{equation}
can be transformed to the fast diffusion equation $u_t=(u^{-1}u_x)_x$
with the point transformation
\begin{equation}\label{Transform14gToFDE}
\tilde t=e^{\varepsilon t}/\varepsilon,\quad \tilde x=x+\varepsilon t,\quad \tilde u=e^{x+\varepsilon t}u.
\end{equation}

The list of known $G^{\max}$-inequivalent Lie solutions of the fast diffusion equation
is exhausted by the following ones~\cite{Popovych&Vaneeva&Ivanova2005}:
\begin{gather}\label{Solutions.for.u-1}
\begin{split}&
1)\ u=\dfrac{1}{1+\delta e^{x+t}},\qquad
2)\ u=e^{x},\qquad
3)\ u=\dfrac{1}{x-t+C te^{-x/t}},\qquad
4)\ u=\dfrac{2t}{x^2+\delta t^2},
\\[1ex]&
5)\ u=\dfrac{2t}{\cos^2x},\qquad
6)\ u=\dfrac{-2t}{\cosh^2x},\qquad
7)\ u=\dfrac{2t}{\sinh^2x},\qquad
8)\ u=\dfrac{2\sin 2t}{\cos2t-\cos 2x},\\[1ex]&
9)\ u=\dfrac{2\sinh 2t}{\cosh 2x-\cosh 2t},\qquad
10)\ u=-\dfrac{2\sinh 2t}{\cosh 2x+\cosh 2t},\\[1ex]&
11)\ u=\dfrac{2\cosh 2t}{\sinh 2x-\sinh 2t},\qquad
12)\ u=\dfrac{2\sin 2t}{\cosh 2x-\cos 2t},\qquad
13)\ u=\dfrac{2\sinh 2t}{\cosh 2t-\cos 2x}.
\end{split}\end{gather}
Here $\delta=\pm1$, $C$ is an arbitrary constant; 1)--7) are Lie invariant solutions and 8)--13) are obtained with nonclassical potential symmetries.

Applying transformation~\eqref{Transform14gToFDE} to~\eqref{Solutions.for.u-1} we obtain a list of exact solutions
of the variable-coefficient equation~\eqref{EqToSolve14g}:
\begin{gather*}
u=\dfrac{e^{-(x+\varepsilon t)}}{1+\delta e^{x+\varepsilon t+e^{\varepsilon t}/\varepsilon}},\qquad
u=\dfrac{\varepsilon e^{-(x+\varepsilon t)}}{\varepsilon x+\varepsilon^2t-e^{\varepsilon t}
+Ce^{\varepsilon t-\varepsilon(x+\varepsilon t)e^{-\varepsilon t}}},
\\[1ex]
u=\dfrac{2e^{-x}}{\varepsilon(x+\varepsilon t)^2+\delta e^{2\varepsilon t}},\qquad
u=\dfrac{2e^{-x}}{\varepsilon\cos^2(x+\varepsilon t)},\qquad u=-\dfrac{2e^{-x}}{\varepsilon\cosh^2(x+\varepsilon t)},
\\[1ex]
u=\dfrac{2e^{-x}}{\varepsilon\sin^2(x+\varepsilon t)},\qquad
u=\dfrac{2e^{-(x+\varepsilon t)}\sin(2e^{\varepsilon t}/\varepsilon)}{\cos(2e^{\varepsilon t}/\varepsilon)-\cos2(x+\varepsilon t)},
\\[1ex]
u=\dfrac{2e^{-(x+\varepsilon t)}\sinh(2e^{\varepsilon t}/\varepsilon)}{\cosh2(x+\varepsilon t)-\cosh(2e^{\varepsilon t}/\varepsilon)},\quad
u=-\dfrac{2e^{-(x+\varepsilon t)}\sinh(2e^{\varepsilon t}/\varepsilon)}{\cosh2(x+\varepsilon t)+\cosh(2e^{\varepsilon t}/\varepsilon)},
\\[1ex]
u=\dfrac{2e^{-(x+\varepsilon t)}\cosh(2e^{\varepsilon t}/\varepsilon)}{\sinh2(x+\varepsilon t)-\sinh(2e^{\varepsilon t}/\varepsilon)},\quad
u=\dfrac{2e^{-(x+\varepsilon t)}\sin(2e^{\varepsilon t}/\varepsilon)}{\cosh2(x+\varepsilon t)-\cos(2e^{\varepsilon t}/\varepsilon)},
\\[1ex]
u=\dfrac{2e^{-(x+\varepsilon t)}\sinh(2e^{\varepsilon t}/\varepsilon)}{\cosh(2e^{\varepsilon t}/\varepsilon)-\cos2(x+\varepsilon t)}.
\end{gather*}

The lists of known exact solutions and Lie reductions of the constant coefficient and some of the variable coefficient equations
can be found, e.g., in~\cite{Ivanova2006DiffusionSolutions,Polyanin&Zaitsev2004}.
Using them one can easily obtain exact solutions of any of the variable coefficient equations having the Lie symmetry algebra
of dimension not less then 4.
Since the solutions of these equations can be reconstructed from
ones presented in other sections, we can turn back to the more interesting cases.

\section{Example of Lie reduction}\label{SectionOnExactSolOfGenBurEq}

In this section we consider in more details Lie reductions of variable coefficient equation~3.\ref{a1bufeax2ghheax2}
\begin{equation}\label{Eqa1bufeax2ghheax2}
e^{px^2}u_t=(e^{px^2}u_x)_x+e^{px^2}uu_x,
\end{equation}
which is invariant with respect to the three-dimensional Lie symmetry algebra
\[
\langle\p_t,\,e^{-2pt}\p_x,\,\p_x-2p\p_u\rangle.
\]
In contrast to the case of equations with four-dimensional Lie symmetry algebra
we cannot reduce equation~\eqref{Eqa1bufeax2ghheax2} to an equation
of form~\eqref{eqDKfgh} with $f=g=h=1$.
However, it is an interesting feature of this equation that using a point transformation
$v=u+2px$ it can be mapped to a constant coefficient reaction--convection--diffusion equation
\begin{equation}\label{EqCCa1bufeax2ghheax2}
v_t=v_{xx}+vv_x-2pv
\end{equation}
that does not belong to class~\eqref{eqDKfgh}.
Similarly to Section~\ref{SectionOnExactSolFromEquiv} for simplification of the technical calculations
we will investigate the constant coefficient equation~\eqref{EqCCa1bufeax2ghheax2} instead of~\eqref{Eqa1bufeax2ghheax2}.
The Lie symmetry algebra of equation~\eqref{EqCCa1bufeax2ghheax2}
\[
\langle X_1=\p_t,\,X_2=e^{-2pt}(\p_x+2p\p_v),\,X_3=\p_x\rangle
\]
is a realization of $A_{2.1}\oplus A_1$~\cite{Mubarakzyanov1963a}.
These operators generate the following group of point transformations:
\[
\tilde t=t+\varepsilon_1,\quad \tilde x=x+\varepsilon_2e^{-2pt}+\varepsilon_3,\quad \tilde v=v+2\varepsilon_2pe^{-2pt}.
\]
A list of proper inequivalent subalgebras of the given algebra is exhausted by the following ones
\[
\langle X_1+\alpha X_3\rangle,\quad \langle X_3+\varepsilon X_2\rangle,\quad \langle X_2\rangle,\quad
\langle X_1,\,X_3\rangle,\quad \langle X_2,\,X_3\rangle,\quad \langle X_1+\beta X_3,\,X_2\rangle,
\]
where $\alpha$ and $\beta$ are arbitrary constants, $\varepsilon=0,\pm1$~\cite{Patera&Winternitz1977}.

The first three (one-dimensional) subalgebras lead to Lie reductions to ordinary differential equations,
the fourth and sixth (two-dimensional) ones yield reductions to algebraic equations.
Lie reductions with respect to these subalgebras are summarized in table~4.
One can easily check that it is impossible to construct a Lie ansatz corresponding to the subalgebra $\langle X_2,\,X_3\rangle$.

{\begin{center}
Table~4. Lie reductions of equation~\eqref{EqCCa1bufeax2ghheax2}.
\\[1.5ex] \footnotesize
\begin{tabular}{|l|l|c|c|l|}
\hline \vspacebefore
N&Subalgebra& Ansatz $v=$& $\omega$ &\hfill {Reduced equation\hfill} \\
\hline \vspacebefore
1&$\langle X_1+\alpha X_3\rangle$ & $\varphi(\omega)$ & $x-\alpha t$ & $\varphi''+(\varphi+\alpha)\varphi'-2p\varphi=0$\\[0.5ex]
 \vspacebefore
2&$\langle X_3+\varepsilon X_2\rangle$ &  $\varphi(\omega)+\dfrac{2p\varepsilon x}{e^{2pt}+\varepsilon}$ & $t$
& $\varphi'=\dfrac{2pe^{2p\omega}}{e^{2p\omega}+\varepsilon}\varphi$\\[0.7ex]
 \vspacebefore
3&$\langle X_2\rangle$ & $\varphi(\omega)+2px$ & $t$
& $\varphi'=0$\\[0.5ex]
 \vspacebefore
4&$\langle X_1,\,X_3\rangle$ & $C$ & --- & $C=0$\\[0.5ex]
 \vspacebefore
5&$\langle X_1+\beta X_3,\,X_2\rangle$ & $2px-2p\beta t+C$ & --- & $-2p\beta=0$\\[0.5ex]
\hline
\end{tabular}
\end{center}

We are succeeded in solving equations 2--5 from table~4 that give us the following invariant solutions of equation~\eqref{EqCCa1bufeax2ghheax2}:
\begin{gather*}
v=0,\quad v=2px+C,\quad v=\dfrac{2p\varepsilon x+Ce^{2pt}}{e^{2pt}+\varepsilon}.
\end{gather*}
The corresponding exact invariant solutions of equation~\eqref{Eqa1bufeax2ghheax2} have the form
\begin{gather}\label{solutions.for.Eqa1bufeax2ghheax2}
u=-2px,\quad u=C,\quad u=\dfrac{2p\varepsilon x+Ce^{2pt}}{e^{2pt}+\varepsilon}-2px,
\end{gather}
where $C$ is an arbitrary constant.

Ansatzes~4.2 and~4.3 give a hint for a possible form
\[
v=\varphi(t)x+\psi(t)
\]
of nonlinear separation of variables for construction of exact solutions of equation~\eqref{EqCCa1bufeax2ghheax2}.
Substitution of the ansatz to equation~\eqref{EqCCa1bufeax2ghheax2} leads to antireduction:
\[
\varphi'=\varphi^2-2p\varphi,\quad \psi'=\varphi\psi-2p\psi.
\]
Solving the above system of ODEs for $\varphi$ and $\psi$ we obtain exactly the solutions
of equations~4.2 and~4.3.

\begin{note}
Using the point transformation $\tilde t=e^{-2pt}$, $\tilde x=x$, $\tilde v=e^{2pt}v$ equation~\eqref{EqCCa1bufeax2ghheax2} can be mapped to
a variable coefficient Burgers equation $\tilde v_{\tilde t}=-2p{\tilde t}^{-1}\tilde v_{\tilde x\tilde x}-2p\tilde v\tilde v_{\tilde x}$
studied in~\cite{Kingston&Sophocleous1991}.
\end{note}

\begin{note}
The well-known Cole--Hopf transformation $v=2w_x/w$ reduces equation~\eqref{EqCCa1bufeax2ghheax2}
to the famous constant coefficient reaction--diffusion equation with weak nonlinearity
\[
w_t=w_{xx}-2pw\ln|w|.
\]
After application of the Cole--Hopf transformation to the list of known exact solutions (see, e.g.,~\cite{Polyanin&Zaitsev2004}) of the equation
with weak nonlinearity we obtain exactly solutions~\eqref{solutions.for.Eqa1bufeax2ghheax2} of equation~\eqref{EqCCa1bufeax2ghheax2}.
\end{note}

In this and previous sections we have constructed exact solutions of variable coefficient equations
of form~\eqref{eqDKfgh} reducible to constant coefficient ones which
either belong or not to class~\eqref{eqDKfgh}. In the next section we consider
an essentially variable coefficient equation that is distinguished with its symmetry properties from class~\eqref{eqDKfgh}.

\section{Group analysis of a remarkable {\mathversion{bold}$sl(2,\mathbb{R})$}-invariant\\ diffusion--convection equation}\label{SectionGrAnOfSl2Eq}

Analyzing the results of group classification, we can observe a number
of $\hat G^{\Equiv}$-inequivalent equations~\eqref{eqDKfgh}
which are invariant with respect to different realizations of the algebra~$sl(2,\mathbb{R})$.
The set of such equations is practically exhausted by
the well-known (``constant coefficient'') Burgers and $u^{-4/3}$-diffusion equations
and by the equations which are equivalent to them with respect to additional transformations
(cases~3.\ref{A1Buf1h1} and~3.\ref{A-43B0f1h}).
This set is supplemented by the unique essentially variable coefficient equation
\begin{equation}\label{A-65B1fx2hx2Copy}
x^2u_t=(u^{-6/5}u_x)_x+x^2u_x
\end{equation}
(case~3.\ref{A-65B1fx2hx2}). The $sl(2,\mathbb{R})$-invariance of~\eqref{A-65B1fx2hx2Copy} is directly connected
with the fact that $h$ is not constant.
The corresponding realization of the algebra~$sl(2,\mathbb{R})$ is quite different from ones of
cases~3.\ref{A1Buf1h1} and~3.\ref{A-43B0f1h} and
is the maximal Lie invariance algebra of equation~\eqref{A-65B1fx2hx2Copy}.
It was the reason for us to study equation~\eqref{A-65B1fx2hx2Copy} from the symmetry point of view in details.
Indeed, instead of equation~\eqref{A-65B1fx2hx2Copy} we investigate the simpler equation
\begin{equation}\label{A-65B1fx2hx2Copy1}
x^2v_t=vv_{xx}-\frac56(v_x)^2+x^2v_x
\end{equation}
for the function $v=u^{-6/5}$, i.e. $u=v^{-5/6}$.
The maximal Lie invariance algebra~$A^{\max}$ of~\eqref{A-65B1fx2hx2Copy1} is generated by the basis operators
\[
P_t=\p_t,\quad D=t\p_t+x\p_x+3v\p_v,\quad \Pi=t^2\p_t+(2tx+x^2)\p_x+6(t+x)v\p_v.
\]
These operators generate the following one-parameter groups of point transformations:
\begin{gather*}
P_t\colon \quad \tilde t=t+\varepsilon,\quad \tilde x=x,\quad \tilde v=v;
\\[1ex]
D\colon \quad \tilde t=e^{\varepsilon}t,\quad \tilde x=e^{\varepsilon}x,\quad  \tilde v=e^{3\varepsilon}v;
\\[.5ex]
\Pi\colon \quad  \tilde t=\dfrac{t}{1-\varepsilon t}, \quad
\tilde x=\dfrac{t+x}{1-\varepsilon(t+x)}-\dfrac{t}{1-\varepsilon t},\quad
\tilde v=\dfrac v{(1-\varepsilon(t+x))^6}.
\end{gather*}
The complete Lie invariance group~$G^{\max}$ is generated by both the above continuous transformations and
the discrete transformation of changing of sign in the triple~$(t,x,v)$.
The transformations from~$G^{\max}$ can be used for construction of new solutions from known ones.

A list of proper $G^{\max}$-inequivalent subalgebras of~$A^{\max}$ is exhausted by the algebras
$\langle P_t\rangle$, $\langle D\rangle$, $\langle P_t+\Pi\rangle$, $\langle P_t,\, D\rangle$.

The simplest Lie reduction of equation~\eqref{A-65B1fx2hx2Copy1} to an algebraic equation
is obtained with the algebra~$\langle P_t, D\rangle$. Namely, the ansatz~$v=Cx^3$ results in the algebraic equation
$C(C-2)=0$ with respect to the constant~$C$, i.e. either $C=0$ or $C=2$. Thus, we have two solutions of
equation~\eqref{A-65B1fx2hx2Copy1} $v=0$ and $v=2x^3$
which are invariant with respect to the translation and scale transformations simultaneously.

The above one-dimensional subalgebras reduce equation~\eqref{A-65B1fx2hx2Copy1} to ordinary differential equations.
Let us list the corresponding ansatzes and reduced equations
as well as some partial exact solutions of the reduced equations:

\medskip

\noindent
$\langle P_t\rangle:\quad v=\varphi(\omega),\quad \omega=x,$
\begin{equation}\label{A-65B1fx2hx2Copy1RedEq1}
\varphi\varphi_{\omega\omega}-\frac56\varphi_{\omega}^2=-\omega^2\varphi_\omega;\qquad
\varphi=C, \quad \varphi=2\omega^3.
\end{equation}
$\langle D\rangle:\quad  v=t^3\varphi(\omega),\quad \omega=\dfrac{x}{t},$
\begin{equation}\label{A-65B1fx2hx2Copy1RedEq2}
\varphi\varphi_{\omega\omega}-\frac56\varphi_{\omega}^2=3\omega^2\varphi-(\omega+1)\omega^2\varphi_{\omega};
\end{equation}
$\varphi=0, \quad \varphi=2\omega^3, \quad \varphi=\dfrac34\omega^4+2\omega^3, \quad \varphi=2\omega^3(\omega+1)^3,
\quad \varphi=\dfrac54\omega^4(\omega+1)^2+2\omega^3(\omega+1)^2$.
\\[2ex]
$\langle P_t+\Pi\rangle:\quad v=((t+x)^2+1)^3\varphi(\omega),\quad \omega=\arctan(t+x)-\arctan t,$
\begin{equation}\label{A-65B1fx2hx2Copy1RedEq3}
\varphi\varphi_{\omega\omega}-\frac56\varphi_{\omega}^2=-\varphi_{\omega}\sin^2\omega-6\varphi^2.
\end{equation}
Here $C$ is an arbitrary constant.

Finally, we have the following set of $G^{\max}$-inequivalent Lie invariant exact solutions
of equation~\eqref{A-65B1fx2hx2Copy1} (below $\delta\in\{0,1\}$):
\begin{gather*}
v=\delta,\quad
v=2x^3,\quad
v=\dfrac{3x^4}{4t}+2x^3, \quad
v=x(t+x)^2\left(\dfrac54\dfrac{x^3}{t^3}+2\dfrac{x^2}{t^2}\right).
\end{gather*}
We can extend the set of solutions with non-trivial invariance transformations:
\begin{gather*}
v=C(t+x)^6,\quad
v=2\dfrac{x^3}{t^3}(t+x)^3,\quad
v=\dfrac{3x^4}{4t}\dfrac{(t+x)^2}{Ct-1}+2\dfrac{x^3}{t^3}(t+x)^3,\\[1ex]
v=x\dfrac{(t+x)^2}{Ct+1}\left(\dfrac54\dfrac{x^3}{t^3}+2\dfrac{x^2}{t^2}(C(t+x)+1)\right).
\end{gather*}

All the constructed solutions are polynomials with respect to~$x$ of degree not greater than 6
with coefficients depending on~$t$, i.e. they have the form
\begin{equation}\label{PolynomXd6}
v=\sum_{i=0}^6 \varphi^i(t)x^i.
\end{equation}
The set of all solutions of the form~\eqref{PolynomXd6} is closed with respect to transformations
from~$G^{\max}$ and is exhausted, up to translations with respect to~$t$ and scale transformations,
by the above solutions $v=\delta$, $v=\delta(t+x)^6$ and the solutions given by the generalized ansatz
\begin{equation}\label{A-65B1fx2hx2SuperAnsatz}
v=2x^3+\varphi^4(t)x^4+\varphi^5(t)x^5+\varphi^6(t)x^6,
\end{equation}
where $\varphi^i$ solve the reduced system
\begin{equation}\label{A-65B1fx2hx2CopyReducedSystem}
\varphi^4_t = 7\varphi^5-\dfrac43(\varphi^4)^2, \quad
\varphi^5_t = 18\varphi^6-\dfrac43\varphi^4\varphi^5, \quad
\varphi^6_t = -\dfrac56(\varphi^5)^2+2\varphi^4\varphi^6.
\end{equation}

Reduction of equation~\eqref{A-65B1fx2hx2Copy1} with ansatz~\eqref{A-65B1fx2hx2SuperAnsatz} to
system~\eqref{A-65B1fx2hx2CopyReducedSystem} is a consequence of generalized nonclassical invariance
of~\eqref{A-65B1fx2hx2Copy1} with respect to the third-order evolutionary operator
\[
(x^3v_{xxx}-12x^2v_{xx}+60xv_x-120v+12x^3)\p_v.
\]
See, e.g.,~\cite{Fokas&Liu1994,Zhdanov1995} for a definition of generalized nonclassical
(or conditional Lie-B\"acklund) symmetries of evolution equations.

System~\eqref{A-65B1fx2hx2CopyReducedSystem} can be reduced to the single third-order ordinary
differential equation for the function~$\varphi^4$:
\begin{equation}\label{A-65B1fx2hx2EqForPhi4}
63\varphi^4_{ttt}+387(\varphi^4_t)^2+126\varphi^4\varphi^4_{tt}+192(\varphi^4)^2\varphi^4_t+16(\varphi^4)^4=0.
\end{equation}
The set of all solutions of the form~\eqref{A-65B1fx2hx2SuperAnsatz} is also closed with respect to transformations
from~$G^{\max}$ and, therefore, $G^{\max}$ naturally induces the invariance groups of
system~\eqref{A-65B1fx2hx2CopyReducedSystem} and equation~\eqref{A-65B1fx2hx2EqForPhi4}.
Indeed, the maximal Lie invariance algebra~$A'$ of~\eqref{A-65B1fx2hx2EqForPhi4} is generated by the operators
\[
P'=\p_t,\quad D'=t\p_t-\varphi^4\p_{\varphi^4},\quad \Pi'=t^2\p_t+(6-2t\varphi^4)\p_{\varphi^4}.
\]
Using one-dimensional subalgebras of~$A'$, we can reduce~\eqref{A-65B1fx2hx2EqForPhi4}
to algebraic equations and, after solving them we construct Lie invariant solutions of~\eqref{A-65B1fx2hx2EqForPhi4}
and, therefore, of~\eqref{A-65B1fx2hx2Copy1} and~\eqref{A-65B1fx2hx2CopyReducedSystem}.
Thus, the subalgebra~$\langle P'\rangle$ leads to the simplest ansatz $\varphi^4=C$
and the algebraic reduced equation $C=0$. The subalgebra~$\langle P'+\Pi'\rangle$ gives
the ansatz $\varphi^4=(6t+C)/(t^2+1)$ and incompatible over~$\mathbb{R}$ algebraic equation $(4C^2+81)(C^2+36)=0$.
The most interesting reduction of such type is connected with the subalgebra~$\langle D'\rangle$.
The $\langle D'\rangle$-invariant ansatz $\varphi^4=C/t$ reduces~\eqref{A-65B1fx2hx2EqForPhi4}
to the algebraic equation $16C^4-192C^3+639C^2-378C=0$ having non-trivial set of solutions $\{0, 3/4, 21/4, 6\}$.
The corresponding invariant solutions of equation~\eqref{A-65B1fx2hx2Copy1} are adduced above.

It is obvious that equation~\eqref{A-65B1fx2hx2Copy} has the nonclassical symmetry operator $\p_t+\p_x$
which gives the ansatz $v=\varphi(\omega)$, $\omega=t+x$ and reduces \eqref{A-65B1fx2hx2Copy} to
the equation $\varphi\varphi_{\omega\omega}-\frac56\varphi_{\omega}^2=0$.
The corresponding solution $v=C_1(x+t+C_0)^6$, where $C_0$ and $C_1$ are arbitrary constants,
is Lie invariant and equivalent to one adduced above.
Let us note that the same operator $\p_t+\p_x$ is a nonclassical symmetry of any equation~\eqref{eqDKfgh},
where $f=h$ and $g=1$.
Detailed discussion of nonclassical (reduction) operators for equations~\eqref{eqDKfgh} is given in the next section.

\section{On nonclassical symmetries}\label{SectionNonclasSym}

Reduction operators (nonclassical symmetries, $Q$-conditional symmetries)
of equations~\eqref{eqDKfgh} have the general form $Q = \tau\p_t + \xi\p_x + \eta\p_u$,
where~$\tau$, $\xi$ and~$\eta$ are functions of $t$, $x$ and $u$, and $(\tau, \xi) \ne (0,0)$.

\begin{definition}\label{DefinitionOfCondSym}
The differential equation~$\mathcal{L}$ is called
\emph{conditionally invariant} with respect to the operator $Q$ if
the relation
$Q_{(r)}L(t,x,u_{(r)})\bigl|_{\mathcal{L}\cap\mathcal{Q}^{(r)}}=0$
holds, which is called the \emph{conditional invariance
criterion}. Then $Q$ is called an operator of \emph{conditional
symmetry} (or $Q$-conditional symmetry, nonclassical symmetry etc)
of the equation~$\mathcal{L}$.
\end{definition}

See, e.g.,~\cite{Zhdanov&Tsyfra&Popovych1999,Popovych&Vaneeva&Ivanova2005} for necessary definitions and properties of nonclassical symmetries.
Since~\eqref{eqDKfgh} is an evolution equation, there are two principally different cases of finding $Q$: $\tau\ne0$ and $\tau=0$.

The problem of classification of nonclassical symmetries with $\tau\ne0$ is not completely solved even in the case of constant coefficient
diffusion--convection equations. However, some particular results are known.
Thus, e.g., constant coefficient diffusion equation with power nonlinearity
\[
u_t=(u^\mu u_x)_x
\]
admits non-classical symmetries only for $\mu=-1/2$~\cite{Gandarias2001}.
Ibid a particular case of such operators $Q=\p_t+12x^{-2}u^{1/2}$ is adduced and the corresponding exact solution
$u=(6tx^{-1}+C_1x^3+C_2x^{-2})^2$ is constructed. Here~$C_1$ and~$C_2$ are arbitrary constants.
The same solution were found previously in~\cite{King1992}.
Conditional symmetry operator $Q=\partial_t-\lambda u^{\mu+1}\partial_x$ for the equation $u_t=(u^\mu u_x)_x+\lambda u^{\mu+1}u_x$
is found in the recent work~\cite{Cherniha&Plyuhin2006}.
Constant coefficient diffusion equation with exponential nonlinearity $u_t=(e^uu_x)_x$ admits two inequivalent
reduction operators: $Q=x\p_t+e^u\p_u$ and $Q=x^2\p_t+2xe^u\p_x+2e^u\p_u$~\cite{Fushchych&Serov&Amerov1992,Fushchych&Serov&Tulupova1993}.

We derive the system of determining equations for the reduction operators with $\tau\ne0$ of equations~\eqref{eqDKfgh}:
\begin{gather}
\xi_{uu}A-\xi_uA_u=0  \label{SysDetEqForCondSymtau1}\\
-\eta_{uu}A^2-\eta_uAA_u+2\xi_{xu}A^2-2\xi\xi_ufA-2\xi_uhAB-\eta AA_{uu}+\eta (A_u)^2=0 \label{SysDetEqForCondSymtau2}\\
-2\eta_{xu}A^2-2\eta_xAA_u-\xi_tfA+2\eta\xi_ufA+\xi_{xx}A^2-2\xi\xi_xfA-\xi_xhAB \nonumber \\
\eta\xi fA_u+\eta hA_uB-\xi^2f_xA-\xi h_xAB-\eta hAB_u=0 \label{SysDetEqForCondSymtau3}\\
\eta_t fA-A^2\eta_{xx}-\eta_xhAB+2\xi_x\eta fA-\eta^2fA_u +\xi\eta f_x A=0\label{SysDetEqForCondSymtau4}
\end{gather}
These equations can be partially integrated.
In particular,~\eqref{SysDetEqForCondSymtau1} implies that
\[\textstyle
\xi = \phi (x,t) \int A +\psi(x,t).
\]
Then,~\eqref{SysDetEqForCondSymtau2} can be written in the form
\[\textstyle
\left( \dfrac {(\eta A)_u}{A} \right)_u=2 \left( \phi_xA-\phi^2f\int A -\phi\psi f -\phi h B\right).
\]

In the case $A=u^n$ and $B=u^m$ ($n \ne -1,-2,-\frac 32$, $m \ne -1,-(n+2)$) the latter two formulas give
\begin{gather*}
\xi =\phi(x,t)u^{n+1}+\psi (x,t),
\\
\eta= \phi_2(x,t) u^{-n} +\psi_2(x,t) u +\frac 1{n+1}\phi_xu^{n+2} -\frac{2(n+1)}{(m+1)(m+n+2)}\phi h u^{m+2}\\
-\frac{2(n+1)}{(n+2)(2n+3)}\phi^2f u^{n+3}-\frac{2(n+1)}{n+2}\phi\psi fu^2.
\end{gather*}
Substituting these to~\eqref{SysDetEqForCondSymtau3} and~\eqref{SysDetEqForCondSymtau4} and solving the derived system we get an example
of nonclassical symmetry for equations~\eqref{eqDKfgh}.
\begin{example}
$A=B=u^{-1/2}$, $\forall f(x),h(x)$, $g(x)=1$:\quad $Q=\p_t+\phi\sqrt{u} \p_x$, where $\phi=\phi(x)$ is an arbitrary solution of equation
\[
\phi''+h(x) \phi' -1/2 f(x) \phi^2=0.
\]
\end{example}


Similarly we found an example of nonclassical symmetry in the case of exponential nonlinearities.
\begin{example}
$A=B=e^u$, $\p_t+\phi(x)e^u\p_x$, where $\phi'-f\phi^2-h\phi =0$.
\end{example}

As well-known, the operators with the vanishing coefficient of $\p_t$ form so-called ``no-go'' case
in study of conditional symmetries of an arbitrary (1 + 1)-dimensional evolution equation since
the problem on their finding is reduced to a single equation which is equivalent to the initial
one (see e.g\cite{Fushchych&Shtelen&Serov&Popovych1992,Popovych1998,Zhdanov&Lahno1998}).
Note that ``no-go'' case has to be treated as impossibility only of exhaustive solving of the problem.
A number of particular examples of reduction operators with $\tau=0$ can be constructed under additional constraints
and then applied to finding exact solutions of the initial equation.
Since the determining equation has more independent variables and, therefore,
more degrees of freedom, it is more convenient often to guess a simple solution or a simple ansatz
for the determining equation, which can give a parametric set of complicated solutions of the
initial equation. Namely, in the case $\tau=0$ we have $\xi\ne0$. Up to usual equivalence of reduction
operators, $\xi$ can be assumed equal to 1, i.e. $Q =\p_x +\eta\p_u$. The conditional invariance criterion
implies the determining equation on the coefficient $\eta$
\begin{gather*}
\eta_tf^2-2\eta\eta_{ux}f A-\eta^2\eta_{uu}f A-2\eta^2\eta_ufA_u+\eta\eta_uf_xA-\eta_{xx}fA-3\eta\eta_xfA_u+\eta_xf_xA \\
-\eta_xfhB-\eta^3fA_{uu}+\eta^2f_xA_u+\eta f_xhB-\eta h_xf B-\eta^2fhB_u=0
\end{gather*}
which is reduced with a non-point transformation to equation~\eqref{eqDKfgh}, where $\eta$ becomes a parameter.
We have found some partial solutions of the determining equations.
\begin{example}
$A=B$,  $\forall f$, $\forall h$, $\eta =\phi(x)A^{-1}$, where $\phi'+h\phi =0$.
\end{example}
\begin{example}
$A=u^n$, $B=u^{2n}$, $f=x^{-2}$, $h=x^{-\frac{2(n+1)}{n+2}}$,  $\eta =\frac 1{(n+2)^2}x^{-\frac{2}{n+2}}u^{1-n}$.
\end{example}

\section{Contractions of solutions}\label{SectionContrSolutions}

Consider, as in Section~\ref{SectionDKfghContractions}, the class~$\{\mathcal{L}^\lambda\}$ of totally nondegenerate systems~$\mathcal{L}^\lambda$:
$L(x,u_{(\prho)},\lambda)=0$
of $l$~differential equations for $m$~unknown functions $u=(u^1,\ldots,u^m)$
of $n$~independent variables $x=(x_1,\ldots,x_n)$,
which are parameterized with the (single numeric) parameter~$\lambda$.

Let for any $\lambda$ from a deleted neighborhood of $\lambda_0$ the system~$\mathcal{L}^\lambda$ be Lie invariant
with respect to  a symmetry algebra~$A^\lambda$,
$\mathcal{L}^\lambda\rightharpoonup\mathcal{L}^{\lambda_0}$ and
$A^\lambda\to A^{\lambda_0}$, $\lambda\to\lambda_0$.
Then, according to the results of Section~\ref{SectionDKfghContractions}, the system~$\mathcal{L}^{\lambda_0}$
is Lie invariant with respect to the algebra~$A^{\lambda_0}$.

If $\omega^\lambda$ is a Lie ansatz corresponding to the algebra~$A^\lambda$,
then, using the theorems on convergence of solutions of differential equations with respect to initial data and parameters
(see, e.g.,~\cite{Coddington&Levinson1955})
one can prove that there exists such $\omega^{\lambda_0}$ that $\omega^\lambda\to\omega^{\lambda_0}$, $\lambda\to\lambda_0$ and
$\omega^{\lambda_0}$ is an Ansatz corresponding to the algebra~$A^{\lambda_0}$.
 Consequently the systems reduced with respect to~$\omega^{\lambda}$ converge to one reduced with~$\omega^{\lambda_0}$.
 In such sense we can talk about {\em contractions of ansatzes and reduced systems}.

Consider an example of contractions of ansatzes, reduced equations and solutions of equations~\eqref{eqDKfgh}.
As it is shown above, under the action of contraction $u=1+\frac{\tilde u}\mu$, $\mu\to+\infty$ diffusion equation with power nonlinearity
\begin{equation}\label{eqCCDEPowerNonl}
u_t=(u^\mu u_x)_x
\end{equation}
(case~3.\ref{AumB0f1h}) invariant with respect to
$A^\mu=\langle Q^\mu_1=\p_t,Q^\mu_2=t\p_t-\mu^{-1}u\p_u,Q^\mu_3=\p_x,Q^\mu_4=x\p_x+2\mu^{-1}u\p_u \rangle$ is reduced to the equation
\begin{equation}\label{eqCCDEExpNonl}
\tilde u_t=(e^{\tilde u}\tilde u_x)_x
\end{equation}
(case~2.\ref{AeuB0f1h}) with exponential nonlinearity being invariant with respect to the Lie symmetry algebra
$A^{\exp}=\langle Q_1=\p_t,Q_2=t\p_t-\p_u,Q_3=\p_x,Q_4=x\p_x+2\p_u \rangle$.

The basis elements of the Lie invariance algebra $A^{\exp}$ can be obtained under the same contraction as $Q^\mu_i\to Q_i$.

All possible inequivalent (with respect to inner automorphisms)
one-dimensional subalgebras of the maximal Lie invariance algebras~$A^\mu$ are exhausted by the ones
listed in Table~5 together with the corresponding ansatzes and the reduced ODEs.

\begin{center}
\renewcommand{\arraystretch}{1.2}
Table~5. Reduced ODEs for~\eqref{eqCCDEPowerNonl}. $\mu\ne0,-1,$ $\alpha\ne0,$ $\varepsilon=\pm1,$ $\delta=\sign t.$
\footnotesize
\begin{tabular}{|l|l|c|c|l|}
\hline \vspacebefore
N&Subalgebra& Ansatz $u=$& $\omega$ &\hfill {Reduced ODE\hfill} \\
\hline
1&$\langle Q^\mu_3\rangle$ & $\varphi(\omega)$ & $t$ & $\varphi'=0$\\
2&$\langle Q^\mu_4\rangle$ & $\varphi(\omega)|x|^{2/\mu}$ & $t$ &
   $\varphi'=2\mu^{-2}(2+\mu)\varphi^{\mu+1}$\\
3&$\langle Q^\mu_1\rangle$ & $\varphi(\omega)$ & $x$ & $(\varphi^{\mu}\varphi')'=0$\\
4&$\langle Q^\mu_2\rangle$ & $\varphi(\omega)|t|^{-1/\mu}$ & $x$
  & $(\varphi^{\mu}\varphi')'=-\delta\mu^{-1}\varphi$\\
5&$\langle Q^\mu_1+\varepsilon Q^\mu_3\rangle$ & $\varphi(\omega)$ & $x-\varepsilon t$ &
     $(\varphi^{\mu}\varphi')'=-\varepsilon\varphi'$\\
6&$\langle Q^\mu_2+\varepsilon Q^\mu_3\rangle$ & $\varphi(\omega)|t|^{-1/\mu}$ & $ x-\varepsilon\ln|t|$
  & $(\varphi^{\mu}\varphi')'=-\delta\varepsilon\varphi'-\delta\mu^{-1}\varphi$\\
7&$\langle Q^\mu_1+\varepsilon Q^\mu_4\rangle$ & $\varphi(\omega)e^{2\varepsilon\mu^{-1} t}$
       & $xe^{-\varepsilon t}$
  &$(\varphi^{\mu}\varphi')'=-\varepsilon\omega\varphi'+2\mu^{-1}\varepsilon\varphi$\\
8&$\langle Q^\mu_2+\alpha Q^\mu_4\rangle$ & $\varphi(\omega)|t|^{(2\alpha-1)/\mu}$ & $ x|t|^{-\alpha}$ &
   $(\varphi^{\mu}\varphi')'=\delta\mu^{-1}(2\alpha-1)\varphi-\delta\alpha\omega\varphi'$\\
\hline
\end{tabular}
\end{center}

For the considered equations the optimal systems of subalgebras of~$A^\mu$ are contracted to an optimal system of subalgebras of~$A^{\exp}$.
Let us note, that in general the question whether optimal systems of subalgebras of the maximal Lie invariance algebras of systems $\mathcal{L}^\lambda$
converge to the optimal subalgebras system of~$\mathcal{L}^{\lambda_0}$, remains open.

Contractions of ansatzes and reduced ODEs in cases~5.1 and~5.3 are obvious. Let us consider in more details contraction of case~5.4.
Under the transformation of variables $u=1+\frac{\tilde u}\mu$ the function $\varphi$ should be changed as $\varphi=1+\frac{\tilde \varphi}\mu$.
The ansatz $\tilde u=\varphi(x)|t|^{-1/\mu}$ can be contracted as follows:
\[
\Big(1+\frac{\tilde u}\mu\Big)^\mu=\Big(1+\frac{\tilde \varphi}\mu\Big)^\mu t^{-1} \to e^{\tilde u}=e^{\tilde \varphi}|t|^{-1}, \quad \mu\to\infty,
\]
Therefore, $\tilde u=\tilde \varphi-\ln|t|$. The similarity variable $\tilde\omega=x$ is not changed under the contraction.
Substituting the derived expressions to the reduced equation~4.4 we obtain
\[
\frac\mu{\mu+1}\left[\left(1+\frac{\tilde \varphi}\mu\right)^{\mu+1}\right]''=-\delta\left(1+\frac{\tilde \varphi}\mu\right).
\]
If now $\mu\to\infty$ we get the reduced ordinary differential equation
\[(e^{\tilde\varphi})''=-\delta
\]
for the target equation~\eqref{eqCCDEExpNonl}.

Similarly one can contract all ansatzes and reduced equations of~\eqref{eqCCDEPowerNonl} to ones for equations with exponential nonlinearity.
The results of these contractions are summarized in Table~6.

\begin{center}
\renewcommand{\arraystretch}{1.2}
Table~6. Reduced ODEs for~\eqref{eqCCDEExpNonl}. $\alpha\ne0,\ \varepsilon=\pm1,\ \delta=\sign t$.
\footnotesize
\begin{tabular}{|l|l|c|c|l|}
\hline \vspacebefore
N&Subalgebra& Ansatz $\tilde u=$& $\omega$ &\hfill {Reduced ODE\hfill} \\
\hline
1&$\langle Q_3\rangle$ & $\varphi(\omega)$ & $t$ & $\varphi'=0$\\
2&$\langle Q_4\rangle$ & $\varphi(\omega)+2\ln|x|$ & $t$ & $\varphi'=2e^{\varphi}$\\
3&$\langle Q_1\rangle$ & $\varphi(\omega)$ & $x$ & $(e^{\varphi})''=0$\\
4&$\langle Q_2\rangle$ & $\varphi(\omega)-\ln|t|$ & $x$ & $(e^{\varphi})''=-\delta$\\
5&$\langle Q_1+\varepsilon Q_3\rangle$ & $\varphi(\omega)$ & $x-\varepsilon t$ &
$(e^{\varphi})''=-\varepsilon\varphi'$\\
6&$\langle Q_2+\varepsilon Q_3\rangle$ & $\varphi(\omega)-\ln|t|$ & $x-\varepsilon\ln|t|$
& $(e^{\varphi})''=-\delta(\varepsilon\varphi'+1)$\\
7&$\langle Q_1+\varepsilon Q_4\rangle$ & $\varphi(\omega)+2\varepsilon t$ & $xe^{-\varepsilon t}$ &
$(e^{\varphi})''=-\varepsilon\omega\varphi'+2\varepsilon$\\
8&$\langle Q_2+\alpha Q_4\rangle$ & $\varphi(\omega)+(2\alpha-1)\ln|t|$ & $x|t|^{-\alpha}$ &
     $(e^{\varphi})''=\delta(-\alpha\omega\varphi'+2\alpha-1)$\\
\hline
\end{tabular}
\end{center}

\section{Conclusion}\label{SectionConclusion}

In this second part of the presented series of papers
(see also~\cite{Ivanova&Popovych&Sophocleous2006Part1,Ivanova&Popovych&Sophocleous2006Part3,Ivanova&Popovych&Sophocleous2006Part4})
we investigate in more detail symmetry properties of class~\eqref{eqDKfgh}.
Namely, considering non-trivial limits of parameterized subclasses of equations
from class~(\ref{eqDKfgh}), which generate contractions of the corresponding maximal Lie invariance algebras,
we introduce the notion of contraction of (systems of) differential equations and consider contractions of
equations and from class~(\ref{eqDKfgh}) and ones of their symmetries.
We also investigate $sl(2,\mathbb{R})$-invariant equation~\eqref{A-65B1fx2hx2Copy}
which is ``essentially variable coefficient'' in the sense that it is not reducible to equations of
form~\eqref{eqDKfgh} with constant values of $f$, $g$ and $h$.

Using similar techniques, we can study
other classes of non-linear evolution equations which are closed to the class under consideration,
e.g. the class of variable coefficient reaction--diffusion equations of the general form
$f(x)u_t=(g(x)A(u)u_x)_x+h(x)B(u)$,
where all denotations coincide with ones in~\eqref{eqDKfgh}.
(See the resent work~\cite{Vaneeva&Johnpillai&Popovych&Sophocleous2006} for its particular case with power nonlinearity.)
However, experience of modern group analysis shows
that extension of circle of problems leads to necessity of modification and enhancement of applied tools.

Another natural direction is generalization of the ``contraction concept" to conservation laws of (systems of) diffusion equations
that is considered in the next part~\cite{Ivanova&Popovych&Sophocleous2006Part3} of this series.

\subsection*{Acknowledgements}

NMI and ROP express their gratitude to the hospitality shown by University of Cyprus
during their visits to the University.
Research of NMI was supported by the Erwin Schr\"odinger Institute for Mathematical Physics (Vienna, Austria) in form of Junior Fellowship
and by the grant of the President of Ukraine for young scientists (project number GP/F11/0061).
Research of ROP was supported by Austrian Science Fund (FWF), Lise Meitner project M923-N13.

\end{document}